\documentclass[twocolumn,aps,pra,notitlepage,superscriptaddress,longbibliography]{revtex4-1}
\usepackage[colorlinks, linkcolor = ForestGreen, anchorcolor = blue, citecolor = magenta]{hyperref}
\usepackage{amsmath,amssymb,amsthm,mathtools,mathrsfs}
\usepackage[table,dvipsnames]{xcolor}
\usepackage{tikz}
\usepackage{adjustbox}
\usepackage{dsfont}
\usepackage[normalem]{ulem}
\usepackage{enumerate}
\usepackage[all]{xy} 
\usepackage{natbib}
\usepackage{babel}
\usepackage{physics}

\theoremstyle{plain}
\newtheorem{theorem}{Theorem}
\newtheorem{lemma}{Lemma}
\newtheorem{proposition}{Proposition}

\theoremstyle{definition}
\newtheorem{definition}{Definition}

\newtheorem{example}{Example}

\newcommand{\bd}{\begin{definition}}
\newcommand{\ed}{\end{definition}}
\newcommand{\bt}{\begin{theorem}}
\newcommand{\et}{\end{theorem}}
\newcommand{\bn}{\begin{proposition}}
\newcommand{\en}{\end{proposition}}
\newcommand{\be}{\begin{equation}}
\newcommand{\ee}{\end{equation}}
\newcommand{\blem}{\begin{lemma}}
\newcommand{\elem}{\end{lemma}}
\newcommand{\bx}{\begin{example}}
\newcommand{\ex}{\end{example}}
\newcommand{\bprf}{\begin{proof}}
\newcommand{\eprf}{\end{proof}}

\DeclareMathAlphabet{\mathpzc}{OT1}{pzc}{m}{it} 
 \DeclareFontFamily{OT1}{pzc}{}
 \DeclareFontShape{OT1}{pzc}{m}{it}{ <-> s*[1.2] pzcmi7t }{}
 \DeclareMathAlphabet{\mathpzc}{OT1}{pzc}{m}{it}
 
\def\C{{{\mathbb C}}}
\newcommand{\<}{\langle}
\renewcommand{\>}{\rangle}
\newcommand{\id}{\operatorname{id}}

\def\N{\mathcal{N}}
\def\E{\mathcal{E}}
\def\H{\mathcal{H}}
\def\F{\mathcal{F}}

\def\P{\mathcal{P}}
\def\Q{\mathcal{Q}}
\def\J{\mathscr{J}}
\def\C{\mathbb{C}}

\begin{document}									
\preprint{APS/123-QED}

\title{Operational time-reversal symmetry for unital qubit channels}
\author{Ouyang Ting}
\author{James Fullwood}
\email{fullwood@hainanu.edu.cn}
\author{Zhen Wu}
\email{wzmath@hainanu.edu.cn}
\affiliation{School of Mathematics and Statistics, Hainan University, Haikou, Hainan Province, 570228, China}

\date{\today}

\begin{abstract}
The Bayesian inverse of a quantum channel $\E$ is a channel $\F$ in the reverse direction of $\E$ that yields time-symmetric correlations for sequential measurements performed on open quantum systems. Such an operational form of time-reversal symmetry for open quantum systems is quite remarkable, as the dynamics of open quantum systems are inherently irreversible due to system-environment interactions. Similar to the Petz map, a Bayesian inverse $\F$ is defined with respect to a fiducial reference state $\rho$ for the channel $\E$. However, Bayesian inverses do not always exist, and it is often a non-trivial task to determine the set of states $\rho$ for which a Bayesian inverse of $\E$ exists. In this work, we solve the general problem of quantum Bayesian inversion for unital channels acting on a single qubit. Our analysis is streamlined by demonstrating that finding a Bayesian inverse for a unital qubit channel may be reduced to finding a Bayesian inverse of a Pauli channel, which is simply a mixture of unitary channels associated with the Pauli matrices. As such, we provide a complete description of when operational time-reversal symmetry is attainable for sequential measurements of a single qubit in the presence of unital noise.
\end{abstract}

	\maketitle

\section{Introduction}

Time-reversal symmetry is a characteristic feature of unitary evolution, which models the dynamics of closed quantum systems. In the real world however virtually all quantum systems are open systems, as they are prone to interaction with their environment. The dynamics of open quantum systems are modeled in quantum theory by quantum channels, which are completely positive trace-preserving (CPTP) maps which are inherently irreversible due to noise. Consequently, time-reversal symmetry is an idealization which almost never holds in practice. 

While many take the viewpoint that there should be no privileged role of measurement in quantum theory, it turns out that the act of measurement is the only way in which one may extract information from a quantum system. Thus quantum theory without measurement is a physical theory without information. As such, when it comes to time-reversal symmetry in the presence of noise, perhaps it is more reasonable to seek an operational form of time-reversal symmetry where the order in which measurements are performed has no bearing on the resulting temporal correlations. 

Sequential measurements of classical random variables satisfy such a form of time-reversal symmetry, which is guaranteed by Bayes' rule. In particular, given random variables $X$ and $Y$, it follows from the Bayes' rule $\mathbb{P}(x)\mathbb{P}(y|x)=\mathbb{P}(y)\mathbb{P}(x|y)$ that $E(X,Y)=E(Y,X)$, where $E(X,Y)$ is the expectation value of the product of a sequential measurement of $X$ followed by $Y$, and similarly for $E(Y,X)$. The quantum Bayes' rule as first introduced in Ref.~\cite{FuPa22a} then restores such an operational form of time-reversal symmetry for sequential measurements of systems of qubits. 

In particular, for a sequential measurement of general Pauli observables $\sigma_{A}$ followed by $\sigma_{B}$ on a system of qubits initially prepared in state $\rho$, we assume that the system evolves according to a quantum channel $\E$ modeling the system-environment interaction between measurements. The quantum Bayes' rule then determines a quantum channel $\F$ such that $\<\sigma_{A},\sigma_{B}\>_{(\rho,\E)}=\<\sigma_{B},\sigma_{A}\>_{(\E(\rho),\F)}$, where $\<\sigma_{A},\sigma_{B}\>_{(\rho,\E)}$ is the expectation value of the product of a sequential measurement of $\sigma_{A}$ followed by $\sigma_{B}$ with respect to the pair $(\rho,\E)$ representing the initial state and the dynamics between measurements, and similarly for $\<\sigma_{B},\sigma_{A}\>_{(\E(\rho),\F)}$. In such a case, the channel $\F$ is said to be the \emph{Bayesian inverse} of $\E$ with respect to the initial state $\rho$.

The dependence of a Bayesian inverse $\F$ on a reference state $\rho$ is formally reminiscent of the Petz recovery map~\cite{Pe88}, which is also defined with respect to a reference state $\rho$. The Petz map is a cornerstone of quantum Bayesian inference, serving as a canonical reversal for quantum channels in the contexts of retrodiction and state recovery~\cite{BuSc21,PaBu22,Aw_2023}. However, while the Petz map is always well-defined for any state-channel pair $(\rho,\E)$, Bayesian inverses do not always exist. Moreover, determining whether or not a Bayesian inverse exists for a given pair $(\rho,\E)$ is often a non-trivial problem. While an in-depth study of Bayesian inverses of amplitude damping channels appears in Ref.~\cite{FuPa24a}, a systemic study of the existence and uniqueness of Bayesian inverses is lacking. 

In this work we address the existence of Bayesian inverses for the case of unital channels acting on a single qubit. Unital channels send the maximally mixed state to the maximally state, and they are often referred to as \emph{bistochastic} channels due to such a symmetry. For $(\rho,\E)$ with $\rho$ the maximally mixed state and $\E$ a unital channel it is known that a Bayesian inverse always exists, and is equal to the Hilbert-Schmidt adjoint $\E^{\dag}$. However, when $\rho$ is not maximally mixed the question of existence of Bayesian inverses for unital channels remains largely unexplored. 

The present study then addresses this gap by characterizing the existence and uniqueness of Bayesian inverses for unital qubit channels when the reference state is not maximally mixed. By leveraging the unitary equivalence between unital qubit channels and Pauli channels, we provide a complete geometric description of the set of states for which a unital qubit channel admits a Bayesian inverse. Our analysis reveals that for a broad class of these channels, the existence of an inverse is restricted to a measure-zero subset of the Bloch ball. This implies that operational time-reversal symmetry is an exceptional property for open system dynamics, occurring only under highly specific state-channel configurations. Such results thus establish precise criteria for identifying operational time-reversal symmetry in unital qubit dynamics, while offering insight into the inherent sensitivity of Bayesian inversion to perturbations of the reference state.

\section{Preliminaries}

We consider a qubit initially prepared in state $\rho$ which is to be measured at times $t_A$ and $t_B>t_A$. We associate Hilbert spaces $\H_A\cong \C^2$ and $\H_B\cong \C^2$ with each measurement, and for $X\in \{A,B\}$, we let $\mathcal{L}(\H_X)$ denote the algebra of linear operators on $\H_X$. Due to system-environment interactions, we assume that the qubit dynamically evolves according to a quantum channel $\E:\mathcal{L}(\H_A)\to \mathcal{L}(\H_B)$ between measurements, which is a completely positive, trace-preserving (CPTP) linear map. If Pauli observables $\sigma_A$ and $\sigma_B$ are measured at times $t_A$ and $t_B$, respectively, then the \emph{two-time expectation value} $\<\sigma_A,\sigma_B\>_{(\rho,\E)}$ of the product of the measurements is given by
\[
\<\sigma_A,\sigma_B\>_{(\rho,\E)}=\Tr(\E(\Pi_{A}^{+}\rho\Pi_{A}^{+})\sigma_B)-\Tr(\E(\Pi_{A}^{-}\rho\Pi_{A}^{-})\sigma_B)\, ,
\] 
where $\Pi_{A}^{\pm}$ are the orthogonal projectors onto $\pm 1$-eigenspaces of $\sigma_A$. We note that the appearance of the subscript $(\rho,\E)$ on the two-time expectation value $\<\sigma_A,\sigma_B\>_{(\rho,\E)}$ is to emphasize its dependence on the initial state and the channel governing the dynamics between measurements.

To formulate a notion of operational time-reversal symmetry for such a sequential measurement scenario we utilize the Jomio\l kowski isomorphism~\cite{Ja72}, which maps a quantum channel $\N:\mathcal{L}(\H_A)\to \mathcal{L}(\H_B)$ to the bipartite operator $\J[\N]$ given by
\[
\J[\N]=(\id_A\otimes\, \N)({\tt SWAP})\, ,
\]  
where ${\tt SWAP}=\sum_{i,j}\dyad{i}{j}\otimes \dyad{j}{i}$. We note that $\J[\N]$ is the partial transpose of the Choi matrix $\mathscr{C}[\N]$, and is therefore not positive in general. A quantum channel $\F:\mathcal{L}(\H_B)\to \mathcal{L}(\H_A)$ in the reverse direction as $\E$ is then said to be a \emph{Bayesain inverse} of $\E$ with respect to the initial state $\rho$ if and only if
 \be \label{eq:BA}
\big\{\mathcal{E}(\rho)\otimes \mathds{1},\mathscr{J}[\mathcal{F}]\big\}=\left\{\mathds{1}\otimes\rho\, ,\mathscr{J}[\mathcal{E}^\dagger]\right\} \, ,
 \ee
 where $\{\cdot,\cdot\}$ denotes the anti-commutator, and the channel $\E^{\dag}:\mathcal{L}(\H_B)\to \mathcal{L}(\H_A)$ is the Hilbert-Schmidt adjoint of $\E$. Eq.~\eqref{eq:BA} will then be referred to as the \emph{quantum Bayes' rule} associated such a sequential measurement scenario. We note that when $\rho$ is of full rank and a Bayesian inverse of $\E$ with respect to $\rho$ exists, it is necessarily unique.
 
While the quantum Bayes' rule \eqref{eq:BA} may initially appear unfamiliar, it emerges naturally from the spatiotemporal viewpoint of quantum information as provided by the mathematical formalism of  pseudo-density matrices and quantum states over time~\cite{FJV15,Liu_2024,HHPBS17,FuPa22,FuPa22a,LiNg23}. In particular, multiplying both sides of the quantum Bayes' rule~\eqref{eq:BA} by $1/2$ yields
 \[
 \F\star \E(\rho)={\tt SWAP}(\E\star \rho){\tt SWAP}\, ,
 \]  
 where the $\star$-product maps a channel $\N$ with initial state $\omega$ to $\N\star \omega=\{\omega\otimes \mathds{1}\, , \J[\N]\}/2$, which is the pseudo-density matrix associated with the dynamics $(\omega,\N)$. The pseudo-density matrix $\N\star \omega$ encodes two-time expectation values in the same manner that a bi-partite density matrix encodes expectation values of measurements of local observables performed on spacelike separated systems. In particular, for Pauli observables $\sigma_i$ and $\sigma_j$ measured at times $t_i$ and $t_j>t_i$, it is straightforward to show
\be \label{2TXEXP67}
\<\sigma_i,\sigma_j\>_{(\omega,\,\N)}=\Tr(\N\star \omega(\sigma_i\otimes \sigma_j))\, ,
\ee
which allows for the treatment of temporal correlations with the same mathematical tools as spatial entanglement. However, while the pseudo-density matrix $\N\star \omega$ is Hermitian and of unit trace it is not necessarily positive, as the negative eigenvalues of a pseudo-density matrix act as a witness to correlations which may not be realized by spacelike separated systems.
 
If $\F:\mathcal{L}(\H_B)\to \mathcal{L}(\H_A)$ is a Bayesian inverse of $\E$ with respect to $\rho$, then it follows from Eq.~\eqref{2TXEXP67} that
\[
\<\sigma_{A},\sigma_{B}\>_{(\rho,\E)}=\<\sigma_{B},\sigma_{A}\>_{(\E(\rho),\F)}\, ,
\]
which we view as an operational form of time-reversal symmetry for sequential measurements. We note that such a time-reversal symmetry is quite remarkable, especially in light of the fact that there are system-environment interactions between measurements. 

In this work, we restrict to the case when $\E$ is a unital channel (i.e., when $\E(\mu)=\mu$, where $\mu=\mathds{1}/2$ is the maximally mixed state of a single qubit), and we perform a systematic analysis of when a Bayesian inverse of $\E$ with respect to $\rho$ exists. Of course the existence of such a Bayesian inverse will depend on properties of both $\E$ and the state $\rho$. 

If $\E$ is in fact a unitary channel, so that $\E(\omega)=U\omega U^{\dag}$ for all $\omega \in \mathcal{L}(\H_A)$ for some unitary operator $U$, then it is straightforward to show that the Bayesain inverse of $\E$ exists with respect to any state $\rho$, and is equal to the unitary channel $\omega\mapsto U^{\dag}\omega U$, which is the Hilbert-Schmidt adjoint $\E^{\dag}$. Furthermore, as mentioned in the introduction, a Bayesian inverse of any unital channel $\E$ with respect to the maximally mixed state $\mu$ always exists, and is also equal to the Hilbert-Schmidt adjoint $\E^{\dag}$. As such, the non-trivial cases to consider are when $\E$ is not a unitary channel and when $\rho$ is not the maximally mixed state. Therefore, the questions we now wish to address are the following:
\begin{itemize}
\item
When $\E$ is not a unitary channel and $\rho$ is not the maximally mixed state, does a Bayesian inverse of $\E$ exist with respect to $\rho$?
\item
Does there exist $\rho$ and $\E$ for which a Bayesian inverse exists which is not the Hilbert-Schmidt adjoint $\E^{\dag}$?
\end{itemize}

The following result will significantly simplify our analysis.

\bn \label{PNXS17}
Suppose $\mathcal{N}$ and $\mathcal{E}$ are unitarily equivalent quantum channels, i.e., there exists unitary channels $\mathcal{U}$ and $\mathcal{V}$ such that $\mathcal{N} = \mathcal{U}\circ \mathcal{E}\circ\mathcal{V}$, and let $\rho$ be an initial state for $\E$. Then $\F$ is a Bayesian inverse of $\E$ with respect to $\rho$ if and only if $\mathcal{M}$ is a Bayesian inverse of $\N$ with respect to $\sigma=\mathcal{V}^{\dag}(\rho)$, where $\mathcal{M} = \mathcal{V}^\dagger \circ\mathcal{F}\circ \mathcal{U}^\dagger$.
\en
\bprf
Suppose $\F$ is a Bayesian inverse of $\E$ with respect to $\rho$, so that the quantum Bayes' rule \eqref{eq:BA} holds, and let $U$ and $V$ be unitary operators such that $\mathcal{U}(\omega)=U\omega U^{\dag}$ and $\mathcal{V}(\zeta)=V\zeta V^{\dag}$. Now since $\mathcal{N} = \mathcal{U}\circ \mathcal{E}\circ\mathcal{V}$, it follows that $\mathcal{N}^{\dag} = \mathcal{V}^{\dag}\circ \mathcal{E}^{\dag}\circ\mathcal{U}^{\dag}$. Therefore, we have
\begin{align*}
\J[\N^{\dag}]&=\sum_{i,j}\dyad{i}{j}\otimes V^{\dag}\E^{\dag}(U^{\dag}\dyad{j}{i}U)V \\
&=(\mathds{1}\otimes V^{\dag})\Big(\sum_{i,j}\dyad{i}{j}\otimes \E^{\dag}(U^{\dag}\dyad{j}{i}U)\Big)(\mathds{1}\otimes V) \\
&=(\mathds{1}\otimes V^{\dag})\Big(\sum_{i,j}U\dyad{i}{j}U^{\dag}\otimes \E^{\dag}(\dyad{j}{i})\Big)(\mathds{1}\otimes V) \\
&=(U\otimes V^{\dag})\J[\E^{\dag}](U^{\dag}\otimes V)\, ,
\end{align*}
where the third equality follows from the fact that 
\[
\sum_{i,j}\dyad{i}{j}\otimes U^{\dag}\dyad{j}{i}U=\sum_{i,j}U\dyad{i}{j}U^{\dag}\otimes \dyad{j}{i}\, .
\]
Moreover, replacing $\E^{\dag}$ in the above calculation by $\F$ yields
\[
\J[\mathcal{M}]=(U\otimes V^{\dag})\J[\F](U^{\dag}\otimes V)\, ,
\]
where $\mathcal{M} = \mathcal{V}^\dagger \circ\mathcal{F}\circ \mathcal{U}^\dagger$. Setting $\chi=\{\mathcal{N}(\sigma)\otimes \mathds{1},\J[\mathcal{M}]\}$, we then have
\begin{align*}
\chi&=\big\{\mathcal{N}(\sigma)\otimes \mathds{1},\J[\mathcal{M}]\big\} \\
&=(U\otimes V^{\dag})\big\{\E(\rho)\otimes \mathds{1},\J[\F]\big\}(U^{\dag}\otimes V) \\
&\overset{\eqref{eq:BA}}=(U\otimes V^{\dag})\big\{\mathds{1}\otimes \rho,\J[\E^{\dag}]\big\}(U^{\dag}\otimes V) \\
&=\big\{\mathds{1}\otimes \sigma, \J[\N^{\dag}]\big\}\, ,
\end{align*}
thus $\mathcal{M}$ is a Bayesian inverse of $\N$ with respect to $\sigma=\mathcal{V}^{\dag}(\rho)$. The proof of the converse is similar.
\eprf

Since a unital qubit channel $\E$ is unitarily equivalent to a Pauli channel~\cite{RSW02}, Proposition~\ref{PNXS17} enables us to reduce to the problem of finding a Bayesian inverse for a unital qubit channel $\E$ to finding a Bayesian inverse for a Pauli channel. The Pauli channel $\mathcal{P}$ associated with a given probability vector $(p_0,\ldots,p_3)$ is the mixture of unitary channels given by
\be \label{PALXC67}
\mathcal{P}(\omega)=\sum_{i=0}^{3}p_i\sigma_i\omega\, \sigma_i \qquad \forall \omega\in \mathcal{L}(\H_A) \, ,
\ee
where $\{\sigma_0,\ldots,\sigma_3\}$ are the Pauli spin matrices. We note that since the Pauli matrices are Hermitian, it follows that the Hilbert-Schmidt adjoint $\P^{\dag}=\P$ for every Pauli channel $\P$.

Pauli channels represent a foundational class of open system dynamics, as they encapsulate bit-flip, phase-flip, and depolarizing processes for a single qubit~\cite{NiCh11}. Their diagonal structure with respect to the Pauli basis makes them particularly amenable to direct calculation, providing a natural setting for analyzing the existence and uniqueness of Bayesian inverses.

\section{The unscathed condition for Pauli channels}

Let $\rho$ be an initial state (i.e., a density matrix) for a Pauli channel $\mathcal{P}$ as given by \eqref{PALXC67}. The state $\rho$ will be referred to as \emph{unscathed} with respect to $\P$ if and only if 
\be \label{CondiPauli}
\P(\rho)=\sigma \rho \, \sigma
\ee
for some Pauli matrix $\sigma\in \{\sigma_0,\ldots,\sigma_3\}$. Note that if $\rho$ is unscathed with respect to $\P$, then the purity of $\rho$ is not altered in any way by $\P$. 

Since Pauli channels are unital, it follows that the maximally mixed state $\mu=\mathds{1}/2$ is necessarily unscathed with respect to every Pauli channel. Moreover, in such a case it is known that the Hilbert-Schmidt adjoint $\P^{\dag}=\P$ is a Bayesian inverse of $\P$ with respect to the unscathed state $\mu$. The following result is then a generalization of this fact. 

\begin{theorem} \label{TMX87}
Let $\rho$ be an initial state for a Pauli channel $\P$. Then the Hilbert-Schmidt adjoint $\P^\dagger=\P$ is a Bayesian inverse of $\P$ with respect to $\rho$ if and only if $\rho$ is unscathed with respect to $\mathcal{P}$.
\end{theorem}

Before proving Theorem~\ref{TMX87}, we first classify the states $\rho$ and Pauli channels $\P$ for which $\rho$ is unscathed with respect to $\P$. For this, suppose $(p_0,\ldots,p_3)$ is the probability vector associated with $\P$ and $(r_1,r_2,r_3)$ is the Bloch vector associated with the state $\rho$, so that $\rho=\frac{1}{2}(\mathds{1}+r_1\sigma_1+r_2\sigma_2+r_3\sigma_3)$.

First, we note that the unscathed condition \eqref{CondiPauli} holds for $\sigma= \sigma_0 = \mathds{1}$ if and only if for all $i$, $j$ and $k$ such that $\{i,j,k\}=\{1,2,3\}$ we have $r_{i}(p_j+p_k)=0$. If we assume that $\rho$ is not the maximally mixed state, then $r_i\neq 0$ for some $i\in \{1,2,3\}$. The condition $r_{i}(p_j+p_k)=0$ then implies that $p_j=p_k=0$, so that the non-zero entries of the probability vector $(p_0,\ldots,p_3)$ belong to the set $\{p_0,p_i\}$. However, since we want to consider channels which are non-unitary (since unitary channels always admit a Bayesian inverse), then in such a case we necessarily have $p_0\neq 0\neq p_i$. It then follows from the above conditions that $r_j=r_k=0$, so that $\rho=\frac{1}{2}(\mathds{1}+r_i\sigma_i)$ for some $-1\leq r_i\leq 1$.

For the case when the unscathed condition \eqref{CondiPauli} holds for $\sigma \in \{\sigma_1,\sigma_2,\sigma_3\}$, we also find similar conditions on the probability vector $(p_0,\ldots,p_3)$ associated with $\P$ and the Bloch vector $(r_1,r_2,r_3)$ associated with $\rho$. We summarize all cases in the following:

\bn \label{UNSCXTH71}
Let $\P$ be a Pauli channel. Then the following statements hold.
\begin{itemize}
\item
If the probability vector $(p_0,\ldots,p_3)$ associated with $\P$ contains precisely two non-zero entries $\{p_0,p_i\}$ or $\{p_j,p_k\}$ with $\{i,j,k\} = \{1,2,3\}$, then the state $\rho=\frac{1}{2}(\mathds{1}+r_i\sigma_i)$ is unscathed with respect to $\P$ for all $-1\leq r_i\leq 1$. Conversely, such $\rho$ are the only states which are unscathed with respect to $\P$.
\item
If the probability vector $(p_0,\ldots,p_3)$ associated with $\P$ contains at least three non-zero entries, then $\rho$ is unscathed with respect to $\P$ if and only if $\rho$ is maximally mixed.
\end{itemize}
\en

We now make use of Proposition~\ref{UNSCXTH71} to give a proof of Theorem~\ref{TMX87}.

\begin{proof}[Proof of Theorem~\ref{TMX87}]
If $\P$ is a unitary channel, then there exists a Pauli matrix $\sigma$ such that $\P(\omega)=\sigma\omega\sigma$ for all $\omega\in \mathcal{L}(\H_A)$, so that every state is unscathed with respect to $\P$. Moreover, a Bayesian inverse of $\P$ exists with respect to every state, and is equal to $\P^{\dag}=\P$. Therefore, the result holds trivially for unitary $\P$. So now assume that $\P$ is not a unitary channel. 

\noindent ($\Rightarrow$) Suppose $\P^\dagger=\P$ is a Bayesian inverse of $\P$ with respect to $\rho=\frac{1}{2}(\mathds{1}+r_1\sigma_1+r_2\sigma_2+r_3\sigma_3)$, and let $\{\lambda_i\}$ be such that $\P(\sigma_i) = \lambda_i \sigma_i$ for $i\in \{1,2,3\}$. Since $\P$ is the Bayesian inverse of the process $\P$ with respect to $\rho$ we have $(\P \circ \P)(\rho) = \rho$, which implies
$(\lambda_s^2-1)r_s = 0\,,\, \forall \, s\in\{1,2,3\}$. If $\rho$ is maximally mixed, it's clear that $\rho$ is unscathed with respect to $\P$. So now consider the case where the Bloch vector $(r_1,r_2,r_3)$ of $\rho$ is non-zero. If the Bloch vector has at least two non-zero entries $r_i\neq 0\neq r_j$, then this implies $|\lambda_i^2| = |\lambda_j^2| = 1$. However, according to Fujiwara-Algoet conditions for the complete positivity of a unital qubit channel~\cite{Fujiwara99},  such conditions only hold when $\P$ is a unitary channel. Therefore, our assumption that $\P$ is non-unitary implies there must exist exactly one non-zero entry $r_i$ in the Bloch vector of $\rho$, which implies $|\lambda_i| = 1$. Furthermore, according to Proposition~\ref{UNSCXTH71}, it follows that in such a case $\rho$ is unscathed with respect to $\P$.

\noindent ($\Leftarrow$) Suppose that $\rho$ is unscathed with respect to $\mathcal{P}$. By Proposition~\ref{UNSCXTH71} we know that if the probability vector associated with $\P$ has at least three non-zero entries then the only unscathed state with respect to $\P$ is the maximally mixed state $\mu$, and in such a case the Hilbert-Schmidt adjoint is a Bayesian inverse of $\P$ with respect to $\mu$. If the probability vector $(p_0,\ldots, p_3)$ associated with $\P$ contains precisely two non-zero entries, Proposition~\ref{UNSCXTH71} reduces the unscathed states into two cases:
\begin{itemize}
    \item[(1)] If $p_0 >0$ and $p_k>0$ for a fixed $k \in \{1,2,3\}$, then the initial state, which is unscathed with respect to $\mathcal{P}$, should be of the form $\rho =\frac{1}{2} (\mathds{1} + r_k\sigma_k)$. 
    \item[(2)] If $p_0 = 0$ and $p_j,p_k>0$ for two distinct indices $j,k\in\{1,2,3\}$, then the initial state must be $\rho = \frac{1}{2}(\mathds{1} + r_l\sigma_l)$, where the index $l$ satisfies $\{j,k,l\} = \{1,2,3\}$.
\end{itemize}
In both cases, a strightforward calculation shows that the Hilbert-Schmidt adjoint $\P^\dagger=\P$ satisfies Eq.~\eqref{eq:BA}, and thus is a Bayesian inverse of $\P$ with respect to $\rho$.
\end{proof}

According to Theorem~\ref{TMX87}, the unscathed condition \eqref{CondiPauli} is a sufficient condition for the existence of a Bayesian inverse, and in such a case the Bayesian inverse is the Hilbert-Schmidt adjoint. The following result establishes when the unscathed condition is also a necessary condition for the existence of a Bayesian inverse.

\bn \label{TNZXS89}
Let $\P$ be a Pauli channel whose probability vector $(p_0,\ldots,p_3)$ contains precisely two non-zero entries. Then a Bayesian inverse of $\P$ exists with respect to a state $\rho$ if and only if $\rho$ is unscathed with respect to $\P$. 
\en  
\bprf
As the $\Leftarrow$ implication follows from Theorem~\ref{TMX87}, it suffices to prove the $\Rightarrow$ implication. We will only consider the case where the probability vector associated with $\P$ is of the form $(p,1-p,0,0)$ with $p\in(0,1)$, as the other cases follow \emph{mutatis mutandis}. So suppose there exists a Bayesian inverse of $\P$ with respect to $\rho=\frac{1}{2}(\mathds{1}+\boldsymbol{r}\cdot \boldsymbol{\sigma})$, and suppose $\boldsymbol{r}=(r_1,r_2,r_3)$. If $\boldsymbol{r} = 0$, the initial state is maximally mixed and hence unscathed with respect to $\P$. So now suppose that $\boldsymbol{r} \neq 0$. If $|r_1| = 1$, the initial state is also unscathed according to Proposition~\ref{UNSCXTH71}, so now suppose $|r_1|<1$. After some detailed calculations, we find that the assumption that a Bayesian inverse exists with respect to $\rho$ implies the inequality
\[
-16p^2(1-p)^2(r_2^2+r_3^2)\Big(1+\sum_{i=1}^3 \lambda_i^2 r_i^2\Big) \geq 0\, ,
\]
which holds if and only if $r_2 = r_3 = 0$. It then follows from Proposition~\ref{UNSCXTH71} that $\rho$ is unscathed with respect to $\P$, as desired.
\eprf

By Proposition~\ref{UNSCXTH71} we know that if $\P$ is a Pauli channel whose probability vector $(p_0,\ldots,p_3)$ contains precisely two non-zero entries, then the unscathed states with respect to $\P$ comprise one of the coordinate axes passing through the Bloch ball. It then follows from Proposition~\ref{TNZXS89} that these states are the only states for which such a Pauli channel admits a Bayesian inverse. Furthermore, Theorem~\ref{TMX87} implies that in such a case the Bayesian inverse is the Hilbert-Schmidt adjoint. We note that since the states which lie along a coordinate axis in the Bloch ball form a subset of measure-zero, it follows that for a fixed Pauli channel with exactly two non-zero entries in its probability vector, the property of admitting Bayesian inverse is highly restrictive. In particular, an initial state chosen at random will almost surely fail to satisfy the unscathed condition.

\section{Bayesian Inverses of General Pauli Channels}

Let $\P$ be a Pauli channel with probability vector $(p_0,\ldots,p_3)$, and let $\rho=\frac{1}{2}(\mathds{1}+r_1\sigma_1+r_2\sigma_2+r_3\sigma_3)$ be an initial state. The Jamio\l kowski matrix of $\P$ is then given by 
\[
\mathscr{J}[\P] = \mathds{1}\otimes \mathds{1} + \sum_{i=1}^3 \lambda_i\, \sigma_i\otimes \sigma_i\, ,
\]
where $\lambda_i\in [-1,1]$ are the eigenvalues of $\P$ associated with $\sigma_i$, so that $\P(\sigma_i)=\lambda_i\sigma_i$ for $i=1,2,3$. If $(p_0,\ldots,p_3)$ is the probability vector associated with $\P$, then $\lambda_i=p_0+p_i-p_j-p_k$, where $\{j,k\}=\{1,2,3\}\setminus \{i\}$. If there exists a $\lambda_i$ such that $|\lambda_i| = 1$, the probability vector of the Pauli channel $\P$ contains at most two non-zero entries. Such cases have been discussed in the previous section; therefore, in this section, we assume that $\lambda_i \in (-1,1)$. Now suppose $\Q$ is a Bayesian inverse of $\P$ with respect to $\rho$, and suppose $\mathscr{J}[\Q] = \sum_{i,j=0}^3 a_{ij}\, \sigma_i\otimes \sigma_j$. The quantum Bayes' rule 
\[
\big\{\P(\rho)\otimes \mathds{1},\mathscr{J}[\mathcal{Q}]\big\}=\left\{\mathds{1}\otimes\rho\, ,\mathscr{J}[\P^\dagger]\right\} 
\]
then yields 
\begin{align*}
    a_{00} & = 1\,,     &a_{i0}  &= 0\,, \qquad \qquad \quad \\
    a_{0i} & = \frac{r_i(1-\lambda_i^2)}{1-S}\,,\qquad  &a_{ij}  &= \lambda_i \delta_{ij} - \lambda_i r_i a_{0j} \, ,
\end{align*}
where $S = \sum_{i=1}^3 \lambda_i^2\, r_i^2$ and the final equation for $a_{ij}$ holds for $i,j=1,2,3$. Since $\Q$ is CPTP if and only if its Choi matrix $\mathscr{C}[\Q]$ is positive, we now derive constraints on $a_{ij}$ which ensure $\mathscr{C}[\Q]\geq 0$. For this, we note that the $\mathscr{C}[\Q]=\mathscr{J}[\Q]^{\Gamma}$, where $\Gamma$ denotes the partial transpose. 

Now let $W=(w_{ij})$ be the matrix given by $w_{ij} = \Tr\big[\mathscr{C}[\Q]\, (\sigma_i\otimes \sigma_j)\big]$, and let $\boldsymbol{v}$ be the 3-vector and $\boldsymbol{R}$ be the $3\times 3$ matrix such that 
\begin{equation}
\begin{aligned}
    W = \frac{1}{1-S} \begin{pmatrix}
    1 & \boldsymbol{v}^\dagger \\
    0 & \boldsymbol{R}
\end{pmatrix}\,.
\end{aligned}
\end{equation}
It then follows from Ref.~\cite{Gamel_2016} that $\mathscr{C}[\Q]\geq 0$ if and only if 
\begin{equation}\label{Condi}
    \begin{aligned}
     \eta(\boldsymbol{v},\boldsymbol{R}) &\leq 3\,, \\
     1-2\det(\boldsymbol{R}) - \eta(\boldsymbol{v},\boldsymbol{R})  &\geq 0\,, \\
      (\eta(\boldsymbol{v},\boldsymbol{R})-1)^2-8\det(\boldsymbol{R})  -4(\norm{\boldsymbol{Rv}}_2^2 + \norm{\text{adj}(\boldsymbol{R})}_2^2) &\geq 0\,, 
\end{aligned}
\end{equation}
where $\eta(\boldsymbol{v},\boldsymbol{R})=\norm{\boldsymbol{v}}_2^2 + \norm{\boldsymbol{R}}_2^2$, $\text{adj}(\boldsymbol{R})$ is the adjugate matrix of $\boldsymbol{R}$, and $\norm{A}_2^2 = \Tr(A^\dagger A)$ is the Schatten $2$-norm. As such, when the inequalities~\eqref{Condi} hold, a Bayesian inverse $\Q$ of $\P$ exists with respect to $\rho$, and the Kraus operators of $\Q$ may be obtained from the Choi matrix $\mathscr{C}[\Q]$. In particular, given a spectral decomposition $\mathscr{C}[\Q] = \sum_i \mu_i \dyad{\varphi_i}{\varphi_i}$, the state $\sqrt{\mu_i} \ket{\varphi_i}$ is the purification of the $i$th Kraus operator $K_i$ of the channel $\Q$.

We then arrive at the following:
\bt
A Pauli channel $\P$ admits a Bayesian inverse with respect to $\rho=\frac{1}{2}(\mathds{1}+r_1\sigma_1+r_2\sigma_2+r_3\sigma_3)$ if and only if the inequalities \eqref{Condi} hold. 
\et

We note that it follows from Theorem~\ref{TMX87} and Proposition~\ref{UNSCXTH71} that if the probability vector associated with a Pauli channel $\P$ has at least three non-zero entries, and $\rho$ is not maximally mixed, then a Bayesian inverse of $\P$ with respect to $\rho$ must necessarily differ from the Hilbert-Schmidt adjoint $\P^{\dag}=\P$. Moreover, such a Bayesian inverse will not be unital. We now give two examples which realize such a scenario. 

\begin{figure}
    \centering
    \includegraphics[width=\linewidth]{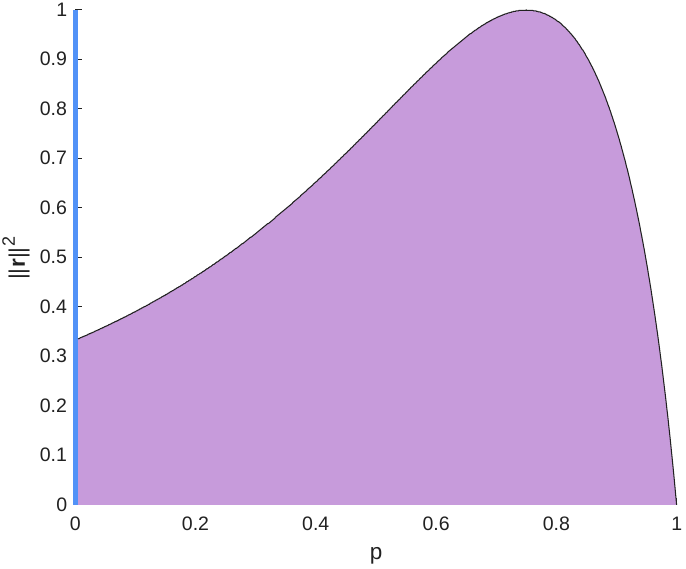}
    \caption{Region of $(p,\norm{\boldsymbol{r}}^2)$-plane for which a Bayesian inverse exists for the depolarizing channel $\mathcal{D}_{p}$ with respect to $\rho=\frac{1}{2}(\mathds{1} +\boldsymbol{r}\cdot \boldsymbol{\sigma})$. The blue line represents the trivial case corresponding to when $\mathcal{D}_{p}$ is the identity channel ($p=0$). The non-trivial cases then correspond to the area colored in violet. In particular, when $p = 0.75$ the channel $\mathcal{D}_{p}$ is completely depolarizing, and a Bayesian inverse exists for all initial states.}
    \label{depolarizing}
\end{figure}

\bx
Let $\mathcal{D}_p$ be the depolarizing channel, which is the Pauli channel given by
\[
\mathcal{D}_p(\omega) = (1-p)\omega + \frac{p}{3}\sum_{i=1}^3 \sigma_i \omega \sigma_i \qquad p\in (0,1)\, .
\]
The eigenvalues of $\mathcal{D}_p$ all coincide, and are equal to $\lambda= 1-\frac{4}{3}p \in [-\frac{1}{3},1)$. Now let $\rho = \frac{1}{2}(\mathds{1} +\boldsymbol{r}\cdot \boldsymbol{\sigma})$ be a fixed initial state, and let $t = \sum_{i=1}^3 r_i^2 = \|\boldsymbol{r}\|^2>0$. In such a case we have
\begin{align*}
     \norm{\boldsymbol{v}}_2^2 & = \frac{1}{(1-S)^2}(1-\lambda^2)^2 t \,, \\
    \norm{\boldsymbol{R}}_2^2 & = \frac{\lambda^2}{(1-S)^2}\Big[(2\lambda^4+1)t^2 - 2(2\lambda^2+1)t + 3\Big] \,, \\
    \norm{\boldsymbol{Rv}}_2^2 & = \frac{1}{(1-S)^4} \lambda^2(1-\lambda^2)^2(1-t)^2 t\,, \\
    \det(\boldsymbol{R}) & = \frac{\lambda^3(t-1)}{1-S}\,, \\
    \norm{\text{adj}(\boldsymbol{R})}_2^2 & = \frac{\lambda^4}{(1-S)^2} \Big[2(1-t)^2+(1-S)^2\Big] \,.
\end{align*}
where $S = \lambda^2 t \in [0,1)$. The inequalities~\eqref{Condi} then yield conditions for which a Bayesian inverse of $\mathcal{D}_p$ exists with respect to $\rho$. In Fig.~\ref{depolarizing}, we plot the region where the inequalities~\eqref{Condi} hold in the $(p,\norm{\boldsymbol{r}}^2)$-plane. In particular, for all $0<p<1$ there exists $\chi(p)$ with $0<\chi(p)\leq 1$ such that $\P$ is Bayesian invertible with with respect to $\rho$ for all $\boldsymbol{r}$ with $0\leq \norm{\boldsymbol{r}}^2\leq \chi(p)$. In such a case, it follows that the Bayesian inverse $\Q$ necessarily differs from the Hilbert-Schmidt adjoint $\P^{\dag}$. We note that for $p=0.75$, i.e., when $\mathcal{D}_p$ is completely depolarizing, a Bayesian inverse $\Q_{\rho}$ exists for every $\rho$, and $\Q_{\rho}$ is the discard-and-prepare channel which discards the input state and outputs the state $\rho$. 
\ex

\begin{figure}
    \centering
    \includegraphics[width=\linewidth]{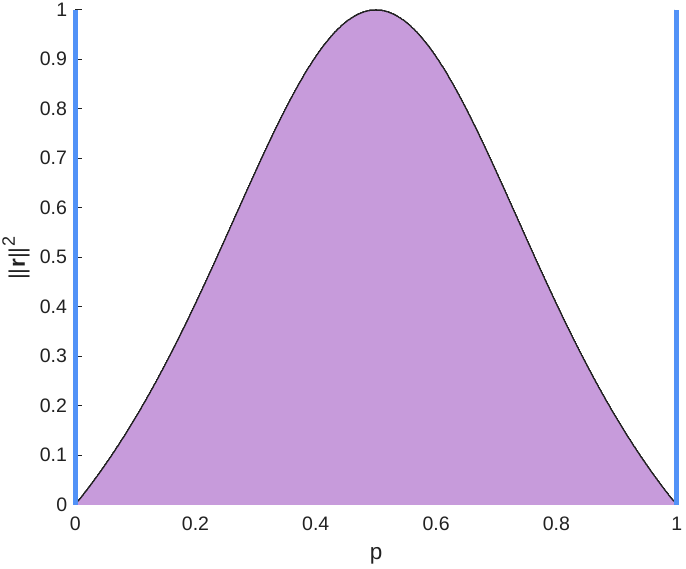}
    \caption{Region of $(p, \norm{\boldsymbol{r}}^2)$-plane for which a Bayesian inverse exists for the BB84 channel $\mathcal{E}_{p,q}$ with respect to the state $\rho=\frac{1}{2}(\mathds{1} +\boldsymbol{r}\cdot \boldsymbol{\sigma})$, where $\boldsymbol{r} = \sqrt{t/3}(1,1,1)$ and $t\in [0,1]$. The blue lines represent the trivial cases corresponding to when $\mathcal{E}_{p,q}$ is the identity channel ($p=0$) and when $\mathcal{E}_{p,q}$ is a $\sigma_2$-unitary channel ($p=1$). The non-trivial cases then correspond to the area colored in violet. In particular, when $p = 0.5$ the channel $\mathcal{E}_{p,q}$ is completely depolarizing, and a Bayesian inverse exists for all initial states.}
    \label{BB84}
\end{figure}

\bx
Let $\mathcal{E}_{p,q}$ be the BB84 channel, which is the Pauli channel given by
\[
\mathcal{E}_{p,q}(\omega) = q^2\,\omega + pq\,\sigma_1 \omega \sigma_1+ p^2\,\sigma_2 \omega \sigma_2 +pq\,\sigma_3 \omega \sigma_3\, ,
\]
where $p\in [0,1]$ and $q=1-p$. Named after the landmark quantum key distribution protocol by Bennett and Brassard~\cite{BB84}, this channel characterizes the fundamental noise encountered in QKD. In particular, it represents the maximum noise level for which Alice and Bob can asymptotically extract a secure key under a given error rate. For an initial state $\rho = \frac{1}{2}(\mathds{1}+\boldsymbol{r}\cdot \boldsymbol{\sigma})$ with a Bloch vector along the diagonal $\boldsymbol{r} = \sqrt{t/3}(1,1,1)$ and $t\in [0,1]$, the inequalities~\eqref{Condi} yield computable conditions for the existence of a Bayesian inverse of $\E_{p,q}$ with respect to $\rho$. In Fig.~\ref{BB84}, we plot the region where the inequalities~\eqref{Condi} hold in the $(p,\norm{\boldsymbol{r}}^2)$-plane. Unlike the case of depolarizing channels, the  region for which Bayesian inverses exist is symmetric about the line $p=0.5$.
\ex

\section{Concluding Remarks}

In this work, we have provided a complete characterization of the existence of Bayesian inverses for unital quantum channels acting on a single qubit. As Bayesian inverses yield an operational form of time-reversal symmetry for sequential measurements performed on an open quantum system, our results establish the precise conditions under which this symmetry is maintained for a qubit subject to unital system-environment interactions between measurements. By exploiting the unitary equivalence between unital and Pauli channels, we demonstrated that the problem of determining the existence of a Bayesian inverse for any unital qubit channel reduces to a tractable analysis of constraints on the Choi matrix of a potential inverse.

We also showed that the Hilbert-Schmidt adjoint acts as a Bayesian inverse of a Pauli channel if and only if the initial state satisfies what we refer to as the \emph{unscathed condition}, which preserves the purity of the initial state. For states that do not satisfy this condition, we have shown that Bayesian inverses can still exist, though they must necessarily differ from the Hilbert-Schmidt adjoint. In the case of depolarizing channels and the BB84 channel, we showed that Bayesian inverses exists for a set of large set of states which are not-unscathed, showing that most Bayesian inverses differ from the Hilbert-Schmidt adjoint in such cases. 

Interestingly, when the Pauli channel associated with a unital channel has precisely three non-zero entries in its probability vector, we were unable to find any examples of Bayesian inverses with respect to states which differ from the maximally mixed state. This is probably due to the fact that either the maximally mixed state is the only state for which a Bayesian inverse exists in such a case, or that such states are restricted to a measure-zero subset of the Bloch ball. Such results underscore the fragility of operational reversibility in the presence of system-environment interactions, while simultaneously highlighting the specific state-channel configurations where temporal correlations can be treated with the same mathematical rigor as spatial entanglement. 

\emph{Acknowledgments}---JF is supported by the Key Development Project of Hainan Province for the project ``Spacetime from Quantum Information", grant no. 126MS0010. 

\bibliography{references}

\end{document}